\documentclass[a4paper,runningheads]{llncs}

\usepackage[english]{babel}
\usepackage{amsfonts}  

\newcommand\vpspace{\ensuremath{\mathsf{VPSPACE}}}
\newcommand\vpspacezero{\ensuremath{\mathsf{VPSPACE}^0}}

\newcommand\cc{\ensuremath{\mathbb{C}}}
\newcommand\rr{\ensuremath{\mathbb{R}}}
\newcommand\qq{\ensuremath{\mathbb{Q}}}
\newcommand\zz{\ensuremath{\mathbb{Z}}}

\newcommand\parc{\ensuremath{\mathsf{PAR}_{\cc}}}
\newcommand\parr{\ensuremath{\mathsf{PAR}_{\rr}}}
\newcommand\pc{\ensuremath{\mathsf{P}_{\cc}}}
\newcommand\pr{\ensuremath{\mathsf{P}_{\rr}}}
\newcommand\parczero{\ensuremath{\mathsf{PAR}_{\cc}^0}}
\newcommand\pczero{\ensuremath{\mathsf{P}_{\cc}^0}}

\newcommand\vp{\ensuremath{\mathsf{VP}}}

\newcommand\vpzeronb{\ensuremath{\vp_{\mathsf{nb}}^0}}
\newcommand\vnp{\ensuremath{\mathsf{VNP}}}

\newcommand\pspace{\ensuremath{\mathsf{PSPACE}}}

\newcommand\ssc{satisfiable sign condition}
\newcommand\unif{\ensuremath{\mathsf{Uniform}}\ }
\newcommand\p{\ensuremath{\mathsf{P}}}
\newcommand\vpnb{\ensuremath{\mathsf{VP}_{\mathsf{nb}}}}
\newcommand\vpnbzero{\ensuremath{\mathsf{VP}^0_{\mathsf{nb}}}}

\newcommand\np{\ensuremath{\mathsf{NP}}}
\newcommand\vparzero{\ensuremath{\mathsf{VPAR}^0}}
\newcommand\poly{\ensuremath{\mathsf{poly}}}
\newcommand\pu{\ensuremath{\mathsf{P}\mbox{-}\mathsf{uniform\ }}}
\newcommand\nc{\ensuremath{\mathsf{NC}}}

\title{VPSPACE and a transfer theorem over the complex field}
\institute{LIP\thanks{UMR 5668 ENS Lyon, CNRS, UCBL, INRIA. Research report
    RR2007-27.},
\'Ecole Normale Sup\'erieure de Lyon.\\
{\tt [Pascal.Koiran,Sylvain.Perifel]@ens-lyon.fr}\\ ~\\
June 2007}

\author{Pascal Koiran\and Sylvain Perifel}
\date{\today}

\begin{document}

\maketitle

\noindent\textbf{Abstract.}
  We extend the transfer theorem of~\cite{KP2006} to the complex field. That
  is, we investigate the links between the class \vpspace\ of families of
  polynomials and the Blum-Shub-Smale model of computation over \cc. Roughly
  speaking, a family of polynomials is in \vpspace\ if its coefficients can be
  computed in polynomial space.  Our main result is that if (uniform,
  constant-free) \vpspace\ families can be evaluated efficiently then the
  class \parc\ of decision problems that can be solved in parallel polynomial
  time over the complex field collapses to \pc.  As a result, one must first
  be able to show that there are \vpspace\ families which are hard to evaluate
  in order to separate \pc\ from $\np_{\cc}$, or even from \parc.

\noindent{\em Keywords:}  computational complexity,
algebraic complexity, Blum-Shub-Smale model, Valiant's model.

\section{Introduction}

In algebraic complexity theory, two main categories of problems are studied:
evaluation and decision problems. The evaluation of the permanent of
a matrix is a typical example of an evaluation problem, and it is well known
that the permanent family is complete for the class \vnp\ of ``easily
definable'' polynomial families~\cite{valiant1979}. Deciding whether a system
of polynomial equations has a solution over \cc\ is a typical example of a
decision problem. This problem is \np-complete in the Blum-Shub-Smale model of
computation over the complex field~\cite{BCSS,BSS1989}.

The main purpose of this paper is to provide a transfer theorem connecting the
complexity of evaluation and decision problems. This paper is therefore in the
same spirit as~\cite{KoiPe06} and~\cite{KP2006} (see
also~\cite{burgisser2001}). In the present paper we work with the class of
polynomial families \vpspace\ introduced in~\cite{KP2006}. Roughly speaking, a
family of polynomials (of possibly exponential degree) is in \vpspace\ if its
coefficients can be evaluated in polynomial space. For instance, it is shown
in~\cite{KP2006} that resultants of systems of multivariate polynomial
equations form a \vpspace\ family. The main result in~\cite{KP2006} was that
if (uniform, constant-free) \vpspace\ families can be evaluated efficiently
then the class \parr\ of decision problems that can be solved in parallel
polynomial time over the real numbers collapses to \pr.

Here we extend this result to the complex field \cc. At first glance the
result seems easier because the order $\leq$ over the reals does not have to
be taken into account. The result of~\cite{KP2006} indeed makes use of a
clever combinatorial lemma of~\cite{grigoriev1999} on the existence of a
vector orthogonal to roughly half a collection of vectors. More precisely, it
relies on the constructive version of this lemma~\cite{CJKPT2006}. On the
complex field, we do not need this construction.

But the lack of an order over \cc\ makes another part of the proof more
difficult. Indeed, over \rr\ testing whether a point belongs to a real variety
is done by testing whether the sum of the squares of the polynomials is zero,
a trick that cannot be used over the complex field. Hence one of the main
technical developments of this paper is to explain how to decide with a small
number of tests whether a point is in the complex variety defined by an
exponential number of polynomials. This enables us to follow the
nonconstructive proof of~\cite{koiran2000} for our transfer theorem.

Therefore, the main result of the present paper is that if (uniform,
constant-free) \vpspace\ families can be evaluated efficiently then the
class \parc\ of decision problems that can be solved in parallel polynomial
time over the complex field collapses to \pc\ (this is precisely stated in
Theorem~\ref{th_transfer}). The class \parc\ plays roughly the same role in
the theory of computation over the complex field as \pspace\ in discrete
complexity theory. In particular, it contains $\np_{\cc}$~\cite{BCSS} (but the
proof of this inclusion is much more involved than in the discrete case). It
follows from our main result that in order to separate \pc\ from $\np_{\cc}$,
or even from \parc, one must first be able to show that there are \vpspace\
families which are hard to evaluate. This seems to be a very challenging lower
bound problem, but it is still presumably easier than showing that the
permanent is hard to evaluate.

{\bf \noindent Organization of the paper}.  We first recall in
Section~\ref{sec_notions} some usual notions and notations concerning
algebraic complexity (Valiant's model, the Blum-Shub-Smale model) and
quantifier elimination. The class \vpspace\ is defined in
Section~\ref{sec_vpspace} and some properties proved in~\cite{KP2006} are
given. Section~\ref{membership} explains how to decide with a polynomial number
of \vpspace\ tests whether a point belongs to a variety. The main difficulty
here is that the variety is given as a union of an exponential number of
varieties, each defined by an exponential number of polynomials. Finally,
Section~\ref{sec_proof} is devoted to the proof of the transfer theorem. Sign
conditions are the main tool in this section. We show that \parc\ problems are
decided in polynomial time if we allow \unif\vpspacezero\ tests. The transfer
theorem follows as a corollary.

\section{Notations and Preliminaries}\label{sec_notions}

\subsection{The Blum-Shub-Smale Model} \label{BSS}

In contrast with boolean complexity, algebraic complexity deals with other
structures than $\{0,1\}$. In this paper we will focus on the complex field
$(\cc,+,-,\times,=)$.
Although the original definitions of
Blum, Shub and Smale~\cite{BSS1989,BCSS} are in terms of 
uniform machines, 
we will follow~\cite{poizat1995} by using families of algebraic circuits to
recognize languages over $\cc$, that is, subsets of
$\cc^{\infty}=\bigcup_{n\geq 0}\cc^n$.

An algebraic circuit
is a directed
acyclic graph whose vertices, called gates, have indegree 0, 1 or 2.  An input
gate is a vertex of indegree 0.  An output gate is a gate of outdegree 0. We
assume that there is only one such gate in the circuit. Gates of indegree 2
are labelled by a symbol from the set $\{+,-,\times\}$. Gates of indegree 1,
called test gates, are labelled ``$= 0$?''. The size of a circuit $C$, in
symbols $|C|$, is the number of vertices of the graph.

A circuit with $n$ input gates computes a function from $\cc^n$ to $\cc$.  On
input $\bar u \in \cc^n$ the value returned by the circuit is by definition
equal to the value of its output gate.  The value of a gate is defined in the
usual way. Namely, the value of input gate number $i$ is equal to the $i$-th
input $u_i$.  The value of other gates is then defined recursively: it is the
sum of the values of its entries for a $+$-gate, their difference for a
$-$-gate, their product for a $\times$-gate.  The value taken by a test gate
is 0 if the value of its entry is $\neq 0$ and 1 otherwise.  Since we are
interested in decision problems, we assume that the output is a test gate: the
value returned by the circuit is therefore 0 or 1.

The class \pc\ is the set of languages $L \subseteq \cc^{\infty}$ such that
there exists a tuple $\bar a \in \cc^p$ and a \p-uniform family of
polynomial-size
circuits $(C_n)$ satisfying the following condition: $C_n$ has exactly
$n+p$ inputs, and for any $\bar x \in \cc^n$, $\bar x \in L \Leftrightarrow
C_n(\bar x,\bar a)=1$. 
The \p-uniformity condition means that $C_n$ can be built in time polynomial
in $n$ by an ordinary (discrete) Turing machine.
Note that $\bar a$ plays the role of the machine
constants of~\cite{BCSS,BSS1989}.

As in~\cite{ChaKoi99}, we define the class \parc\ 
as the set of languages over \cc\ recognized by a 
\pspace-uniform (or equivalently \p-uniform)
family of algebraic circuits of polynomial depth (and possibly exponential
size), with constants $\bar a$ as for \pc. 
Note at last that we could also define similar classes without constants
$\bar a$. We will use the superscript 0 to denote these constant-free classes,
for instance $\pc^0$ and $\parczero$.

We end this section with a theorem on the first-order theory of the complex
numbers: quantifiers can
be eliminated without much increase of the coefficients and degree of the
polynomials. We give a weak version of the result of~\cite{FGM1990}: in
particular, we do not need efficient elimination algorithms. Note that the
only allowed constants in our formulae are 0 and 1 (in particular, only
integer coefficients can appear). For notational consistency with the
remainding of the paper, we denote by $2^s$, $2^d$ and $2^{2^M}$ the number of
polynomials, their degree and the absolute value of their coefficients
respectively. This will simplify the calculations and emphasize that $s$, $d$
and $M$ will be polynomial. Note furthermore that the polynomial $p(n,s,d)$ in
the theorem is independent of the formula $\phi$.
\begin{theorem}\label{th_elimination}
  Let $\phi$ be a first-order formula over 
  $(\cc,0,1,+,-,\times,=)$ of the form $\forall\bar x\psi(\bar
  x)$, where $\bar x$ is a tuple of $n$ variables and $\psi$ a
  quantifier-free formula where $2^s$ polynomials occur. Suppose that their
  degrees are bounded by $2^d$ and their coefficients by $2^{2^M}$ in
  absolute value.

  There exists a polynomial $p(n,s,d)$, independent of $\phi$,
  such that the formula $\phi$ is equivalent to a quantifier-free formula
  $\psi$ in which all polynomials have degree less than $D(n,s,d) =
  2^{p(n,s,d)}$, and their coefficients are integers strictly bounded in
  absolute value by $2^{2^MD(n,s,d)}$.
\end{theorem}

\subsection{Valiant's Model}

In Valiant's model, one computes polynomials instead of recognizing languages.
We thus use arithmetic circuits instead of algebraic circuits. A book-length
treatment of this topic can be found in~\cite{burgisser2000}.

An arithmetic circuit is the same as an algebraic circuit but test gates are
not allowed. That is to say we have indeterminates $x_1,\dots,x_{u(n)}$ as
input together with arbitrary constants of \cc; there are $+$, $-$ and
$\times$-gates, and we therefore compute multivariate polynomials.

The polynomial computed by an arithmetic circuit is defined in the usual way
by the polynomial computed by its output gate. Thus a family $(C_n)$ of
arithmetic circuits computes a family $(f_n)$ of polynomials,
$f_n\in\cc[x_1,\dots,x_{u(n)}]$. The class \vpnb\ defined in \cite{malod2003}
is the set of families $(f_n)$ of polynomials computed by a family $(C_n)$ of
polynomial-size arithmetic circuits, i.e., $C_n$ computes $f_n$ and there
exists a polynomial $p(n)$ such that $|C_n|\leq p(n)$ for all $n$. We will
assume without loss of generality that the number $u(n)$ of variables is
bounded by a polynomial function of $n$. The subscript $\mathsf{nb}$ indicates
that there is no bound on the degree of the polynomial, in contrast with the
original class \vp\ of Valiant where a polynomial bound on the degree of the
polynomial computed by the circuit is required. Note that these definitions
are nonuniform.  The class \unif\vpnb\ is obtained by adding a condition of
polynomial-time uniformity on the circuit family, as in Section~\ref{BSS}.

We can also forbid constants from our arithmetic circuits in unbounded-degree
classes, and define constant-free classes. The only constant allowed is 1 (in
order to allow the computation of constant polynomials). As for classes of
decision problems, we will use the superscript 0 to indicate the absence of
constant: for instance, we will write \vpnbzero\ (for bounded-degree classes,
we are to be more careful: the ``formal degree'' of the circuits comes into
play, see~\cite{malod2003,MP2006}).

\section{The Class VPSPACE}\label{sec_vpspace}

The class \vpspace\ was introduced in \cite{KP2006}. Some of its properties
are given there and a natural example of a \vpspace\ family coming from
algebraic geometry, namely the resultant of a system of polynomial equations,
is provided. In this section, after the definition we give some properties
without proof and refer to~\cite{KP2006} for further details.

\subsection{Definition}

We fix an arbitrary field $K$. The definition of \vpspace\ will be stated in
terms of \emph{coefficient function}. A monomial $x_1^{\alpha_1}\cdots
x_n^{\alpha_n}$ is encoded in binary by $\alpha = (\alpha_1,\dots,\alpha_n)$
and will be written $\bar x^{\alpha}$.

\begin{definition}
  Let $(f_n)$ be a family of multivariate polynomials with integer
  coefficients. The coefficient function of $(f_n)$ is the function $a$ whose
  value on input $(n,\alpha,i)$ is the $i$-th bit $a(n,\alpha,i)$ of the
  coefficient of the monomial $\bar x^\alpha$ in $f_n$. Furthermore,
  $a(n,\alpha,0)$ is the sign of the coefficient of the monomial $\bar
  x^{\alpha}$. Thus $f_n$ can be written as
  $$f_n(\bar x) = \sum_{\alpha}\Bigl((-1)^{a(n,\alpha,0)}\sum_{i\geq 1}
  a(n,\alpha,i)2^{i-1}\bar x^{\alpha}\Bigr).$$
\end{definition}

The coefficient function is a function $a:\{0,1\}^*\rightarrow\{0,1\}$ and can
therefore be viewed as a language. This allows us to speak of the complexity
of the coefficient function. Note that if $K$ is of characteristic $p>0$, then
the coefficients of our polynomials will be integers modulo $p$ (hence with a
constant number of bits). In this paper, we will focus only on the field \cc\
(which is of characteristic 0).

\begin{definition}
  The class \unif\vpspacezero\ is the set of all families $(f_n)$ of
  multivariate polynomials $f_n\in K[x_1,\dots,x_{u(n)}]$ satisfying the
  following requirements:
  \begin{enumerate}
  \item the number $u(n)$ of variables is polynomially bounded;
  \item the polynomials $f_n$ have integer coefficients;
  \item the size of the coefficients of $f_n$ is bounded by $2^{p(n)}$ for
    some polynomial $p$;
  \item the degree of $f_n$ is bounded by $2^{p(n)}$ for some polynomial $p$;
  \item the coefficient function of $(f_n)$ is in \pspace.
  \end{enumerate}
\end{definition}

We have chosen to present only \unif\vpspacezero, a uniform class without
constants, because this is the main object of study in this paper. In keeping
with the tradition set by Valiant, however, the class \vpspace\ is nonuniform
and allows for arbitrary constants. See~\cite{KP2006} for a precise definition.

\subsection{An Alternative Characterization and Some Properties}

Let \unif\vparzero\ be the class of families of polynomials computed by a
\pspace-uniform family of constant-free arithmetic circuits of polynomial
depth (and possibly exponential size). This in fact characterizes
\unif\vpspacezero. The proof is given in~\cite{KP2006}.

\begin{proposition}\label{prop_alternate}
The two classes $\unif\vpspacezero$ and $\unif\vparzero$ are equal.
\end{proposition}

We see here the similarity with \parc, which by definition are those languages
recognized by uniform algebraic circuits of polynomial depth. But of course
there is no test gate in the arithmetic circuits of \unif\vparzero.

We now turn to some properties of \vpspace. The following two propositions
come from~\cite{KP2006}. They stress the unlikeliness of the hypothesis that
\vpspace\ has polynomial-size circuits.

\begin{proposition}\label{prop_equivalence}
  Assuming the generalized Riemann hypothesis (GRH), $\vpnb=\vpspace$ if and
  only if $[\p/\poly=\pspace/\poly\mbox{ and }\vp=\vnp]$.
\end{proposition}

\begin{proposition}
  $\unif\vpspacezero = 
\unif\vpnbzero\Longrightarrow\pspace=\pu\nc$.
\end{proposition}

\begin{remark}
  To the authors' knowledge, the separation ``$\pspace\neq\pu\nc$'' is not
  known to hold (by contrast, $\pspace$ can be separated from logspace-uniform
  \nc\ thanks to the space hierarchy theorem).
\end{remark}

Let us now state the main result of this paper.
\begin{theorem}[main theorem]\label{th_transfer}
  If $\unif\vpspacezero=\unif\vpzeronb$ then $\parczero=\pczero$.
\end{theorem}

Note that the collapse of the constant-free class
\parczero\ to
\pczero\ implies $\parc=\pc$:
just replace constants by new variables so as to transform a \parc\
problem into a \parczero\ problem, and then replace these variables 
by their original values so as to transform a $\pczero$ problem 
into a \pc\ problem.

The next section is devoted to the problem of testing whether a point belongs
to a variety.  This problem is useful for the proof of the theorem: indeed,
following~\cite{koiran2000}, several tests of membership to a variety will be
made; the point here is to make them constructive and efficient. The main
difficulty is that the variety can be defined by an exponential number of
polynomials.

\section{Testing Membership to a Union of Varieties}\label{membership}

In this section we explain how to perform in \unif\vpspacezero\ membership
tests of the form ``$\bar x\in V$'', where $V\subseteq\cc^n$ is a variety. We
begin in Section~\ref{section_variety} by the case where $V$ is given by $s$
polynomials. In that case, we determine after some precomputation whether
$\bar x\in V$ in $n+1$ tests. We first need two lemmas given below in order to
reduce the number of polynomials and to replace transcendental elements by
integers.

Then, in Section~\ref{section_union}, we deal with the case where $V$ is given
as a union of an exponential number of such varieties, as in the actual tests
of the algorithm of Section~\ref{sec_proof}. Determining whether $\bar x\in V$
still requires $n+1$ tests, but the precomputation is slightly heavier.

Let us first state two useful lemmas. Suppose a variety $V$ is defined by
$f_1,\dots,f_s$, where $f_i\in\zz[x_1,\dots,x_n]$. We are to determine whether
$\bar x\in V$ with only $n+1$ tests, however big $s$ might be. In a
nonconstructive manner, this is possible and relies on the following classical
lemma already used (and proved) in~\cite{koiran2000}: any $n+1$ ``generic''
linear combinations of the $f_i$ also define $V$ (the result holds over any
infinite field but here we need it only over \cc). We state this lemma
explicitly since we will also need it in our constructive proof.

\begin{lemma}\label{lemma_variety}
  Let $f_1,\dots,f_s\in\zz[x_1,\dots,x_n]$ be polynomials and $V$ be the
  variety of $\cc^n$ they define. Then for all coefficients
  $(\alpha_{i,j})_{i=1..s,j=1..n+1}\in\cc^{s(n+1)}$ algebraically independent
  over \qq, the $n+1$ linear combinations $g_j=\sum_{i=1}^s\alpha_{i,j}f_i$ (for
  $j$ from 1 to $n+1$) also define $V$.
\end{lemma}

Unfortunately, in our case we cannot use transcendental numbers and must
replace them by integers. The following lemma from~\cite{koiran1997} asserts
that integers growing sufficiently fast will do. Once again, this is a weaker
version adapted to our purpose.

\begin{lemma}\label{lemma_integers}
  Let $\phi(\alpha_1,\dots,\alpha_r)$ be a quantifier-free first-order formula
  over the structure $(\cc,0,1,+,-,\times,=)$, containing only
  polynomials of degree less than $D$ and whose coefficients are integers of
  absolute value strictly bounded by $C$. Assume furthermore that
  $\phi(\bar\alpha)$ holds for all coefficients
  $\bar\alpha=(\alpha_1,\dots,\alpha_r)\in\cc^r$ algebraically independent
  over \qq.

  Then $\phi(\bar\beta)$ holds for any sequence $(\beta_1,\dots,\beta_r)$ of
  integers satisfying $\beta_1\geq C$ and $\beta_{j+1}\geq CD^j\beta_j^D$ (for
  $1\leq j\leq r-1$).
\end{lemma}

The proof can be found in~\cite[Lemma~5.4]{koiran1997} and relies on the lack
of big integer roots of multivariate polynomials.

Let us sketch a first attempt to prove a constructive version of
Lemma~\ref{lemma_variety}, namely that $n+1$ polynomials with integer
coefficients are enough for defining $V$ (this first try will not work but
gives the idea of the proof of the next section). The idea is to use
Lemma~\ref{lemma_integers} with the formula $\phi(\bar\alpha)$ that tells us
that the $n+1$ linear combinations of the $f_i$ with $\alpha_{i,j}$ as
coefficients define the same variety as $f_1,\dots,f_s$.  At first this
formula is not quantifier-free, but over $\cc$ we can eliminate quantifiers
while keeping degree and coefficients reasonably small thanks to
Theorem~\ref{th_elimination}. Lemma~\ref{lemma_variety} asserts that
$\phi(\bar\alpha)$ holds as soon as the $\alpha_{i,j}$ are algebraically
independent. Then Lemma~\ref{lemma_integers} tells us that $\phi(\bar\beta)$
holds for integers $\beta_{i,j}$ growing fast enough. Thus $V$ is now defined
by $n+1$ linear combinations of the $f_i$ with integer coefficients.

In fact, this strategy fails to work for our purpose because the coefficients
involved are growing too fast to be computed in polynomial space. That is why
we will proceed by stages in the proofs below: we adopt a divide-and-conquer
approach and use induction.

\subsection{Tests of Membership}\label{section_variety}

The base case of our induction is the following lemma, whose proof is sketched
in the end of the preceding section. We only consider here a small number of
polynomials, therefore avoiding the problem of too big coefficients mentioned
in the preceding section.
\begin{lemma}\label{lemma_split_two}
  There exists a polynomial $q(n,d)$ such that, if $V\subseteq\cc^n$ is a
  variety defined by $2(n+1)$ polynomials
  $f_1,\dots,f_{2(n+1)}\in\zz[x_1,\dots,x_n]$ of degree $\leq 2^d$ and of
  coefficients bounded by $2^{2^M}$ in absolute value, then:
  \begin{enumerate}
  \item the variety $V$ is defined by $n+1$ polynomials
    $g_1,\dots,g_{n+1}\in\zz[x_1,\dots,x_n]$ of degree $\leq 2^d$ and of
    coefficients bounded by $2^{2^{M+q(n,d)}}$
    in absolute value;
  \item furthermore, the coefficients of the $g_i$ are bitwise computable from
    those of the $f_j$ in working space $Mq(n,d)$.
  \end{enumerate}
\end{lemma}
\begin{proof}
  The first-order formula $\phi(\bar\alpha)$ (where
  $\bar\alpha\in\cc^{2(n+1)^2}$), expressing that the $n+1$ linear
  combinations of the $f_j$'s with coefficients $\bar\alpha$ also define $V$,
  can be written as follows:
  $$\phi(\bar\alpha)\equiv \forall x\in\cc^n\left(
    \bigwedge_{i=1}^{n+1} \sum_{j=1}^{2(n+1)}\alpha_{i,j}f_j(x)=0 \leftrightarrow
    \bigwedge_{j=1}^{2(n+1)} f_j(x)=0\right),$$
  where $\alpha_{i,j}$ is a shorthand for $\alpha_{2(i-1)(n+1)+j}$.
  The polynomials in this formula are of degree $\leq 1+2^d$ and their
  coefficients are bounded in absolute value by $2^{2^M}$.

  Over $\cc$, the quantifier of this formula can be eliminated by
  Theorem~\ref{th_elimination}: $\phi(\bar\alpha)$ is equivalent to a
  quantifier-free formula $\psi(\bar\alpha)$, the polynomials occuring in
  which have their degree less than $D=D(n,\log(3(n+1)),d+1)$ and their
  coefficients strictly bounded in absolute value by $C = 2^{2^MD}$, where
  $D(n,\log(3(n+1)),d+1)=2^{p(n,\log(3(n+1)),d+1)}$ is defined in
  Theorem~\ref{th_elimination}.

  By Lemma~\ref{lemma_variety}, $\psi(\bar\alpha)$ holds for all
  coefficients $\bar\alpha$ algebraically independent, so that we wish to
  apply Lemma~\ref{lemma_integers} with integers $\beta_i$ growing
  sufficiently fast. Let $r=(1+2(n+1)^2)p(n,\log(3(n+1)),d+1)$, so that
  $$D\leq 2^r\mbox{ and } CD^{2(n+1)^2}\leq 2^{2^{M+r}}$$
  and define
  $$\beta_i=2^{2^{M+2ir}}\mbox{ for }1\leq i\leq 2(n+1)^2.$$
  Note that for all $i$, $\beta_i\leq 2^{2^{M+4(n+1)^2r}}$, and it is
  furthermore easy to check that $\beta_1\geq C$ and $\beta_{i+1}\geq
  CD^i\beta_i^D$. Thus by Lemma~\ref{lemma_integers}, $\psi(\bar\beta)$ is
  true. Define the polynomial $q(n,d)=1+4(n+1)^2r$ (up to a
  multiplicative constant for the space complexity below).
  Now, letting
  $$g_i=\sum_{j=1}^{2(n+1)}\beta_{i,j} f_j,$$
  where $\beta_{i,j}$ is a shorthand for $\beta_{2(i-1)(n+1)+j}$,
  proves the first point of the lemma.

  For the second point, remark that the coefficients $\beta_i$ are bitwise
  computable in space $O(M+rn^2)$ and that the coefficients of the $g_i$ are
  merely a sum of $2(n+1)$ products of $\beta_j$ and coefficients of the
  $f_k$. This multiplication uses only space $O(M+rn^2)$ since the integers
  involved have encoding size $2^{O(M+rn^2)}$ (in our case this is
  particularly easy because the $\beta_j$ are powers of 2). The $2n+1$
  additions are also performed in space $O(M+rn^2)$. This
  proves the second point of the lemma.\qed
\end{proof}

Proposition~\ref{prop_membership} now follows by induction.

\begin{proposition}\label{prop_membership}
  There exists a polynomial $p(n,s,d)$ such that, if $V$ is a variety defined
  by $2^s$ polynomials $f_1,\dots,f_{2^s}\in\zz[x_1,\dots,x_n]$ of degree
  $\leq 2^d$ and of coefficients bounded by $2^{2^M}$ in absolute value, then:
  \begin{enumerate}
  \item the variety $V$ is defined by $n+1$ polynomials
    $g_1,\dots,g_{n+1}\in\zz[x_1,\dots,x_n]$ of degree $\leq 2^d$ and of
    coefficients bounded by $2^{2^{M+p(n,s,d)}}$ in absolute value;
  \item moreover, the coefficients of the $g_i$ are bitwise computable from
    those of the $f_j$ in working space $Mp(n,s,d)$.
  \end{enumerate}
\end{proposition}
\begin{proof}
  This is done by induction on $s$. Take $p(n,s,d)=sq(n,d)$ where $q(n,d)$ is
  the polynomial defined in Lemma~\ref{lemma_split_two}. The base case
  $2^s\leq 2(n+1)$ follows from Lemma~\ref{lemma_split_two}. Suppose therefore
  that $2^s > 2(n+1)$. Call $V_1$ and $V_2$ the varieties defined respectively
  by $f_1,\dots,f_{2^{s-1}}$ and by $f_{2^{s-1}+1},\dots,f_{2^s}$. Then
  $V=V_1\cap V_2$ and by induction hypothesis, $V_1$ and $V_2$ are both
  defined by $n+1$ polynomials of degree $\leq 2^d$ whose coefficients are
  bounded by $2^{2^{M+(s-1)q(n,d)}}$ in absolute value and computable in space
  $M(s-1)q(n,d)$.

  Therefore by Lemma~\ref{lemma_split_two}, $V$ is defined by $n+1$
  polynomials of degree $\leq 2^d$ whose coefficients are
  bounded by $2^{2^{M+sq(n,d)}}$ in absolute value and computable in
  space $Msq(n,d)$ as claimed in the proposition.\qed
\end{proof}

\subsection{Union of Varieties}\label{section_union}

In our case, however, the tests made by the algorithm of
Section~\ref{sec_proof} are not exactly of the form studied in the previous
section: instead of a single variety given by $s$ polynomials, we have to
decide ``$x\in W?$'' when $W\subseteq\cc^n$ is the union of $k$ varieties. Of
course, since the union is finite $W$ is also a variety, but the encoding is
not the same as above: now, $k$ sets of $s$ polynomials are given.

A first naive approach is to define $W=\cup_i V_i$ by the different products
of the polynomials defining the $V_i$, but it turns out that there are too
many products to be dealt with. Instead, we will adopt a divide-and-conquer
scheme as previously.
\begin{lemma}\label{lemma_union_two}
  There exists a polynomial $q(n,d)$ such that, if $V_1$ and $V_2$ are two
  varieties of $\cc^n$, each defined by $n+1$ polynomials in
  $\zz[x_1,\dots,x_n]$, respectively $f_1,\dots,f_{n+1}$ and
  $g_1,\dots,g_{n+1}$, of degree $\leq 2^d$ and of coefficients bounded by
  $2^{2^M}$ in absolute value, then:
  \begin{enumerate}
  \item the variety $V=V_1\cup V_2$ is defined by $n+1$ polynomials
    $h_1,\dots,h_{n+1}$ in $\zz[x_1,\dots,x_n]$ of degree $\leq 2^{d+1}$ and
    of coefficients bounded by $2^{2^{M+q(n,d)}}$ in absolute value;
  \item the coefficients of the $h_i$ are bitwise computable from
    those of the $f_j$ and $g_k$ in space $Mq(n,d)$.
  \end{enumerate}
\end{lemma}
\begin{proof}
  The variety $V$ is defined by the $(n+1)^2$ polynomials $f_ig_j$ for $1\leq
  i,j\leq n+1$: these polynomials have degree $\leq 2^{d+1}$. Note moreover
  that there are at most $2^{n(d+1)}$ monomials of fixed degree $\delta\leq
  2^{d+1}$, therefore the coefficients of the $f_ig_j$ are a sum of at most
  $2^{n(d+1)}$ products of integers of encoding size $2^M$. Thus they are
  computable in space $O(Mnd)$ from those of the $f_i$ and $g_j$. This also
  shows that the coefficients of the products $f_ig_j$ are bounded in absolute
  value by $2^{n(d+1)}2^{2^{M+1}}\leq 2^{2^{M+1+n(d+1)}}$. Applying
  Proposition~\ref{prop_membership} now enables to conclude if we take
  $q(n,d)=1+n(d+1)+p(n,\log((n+1)^2),d+1)$, where $p$ is the polynomial
  defined in Proposition~\ref{prop_membership}.\qed
\end{proof}

The next proposition now follows by induction.
\begin{proposition}\label{prop_union2}
  There exists a polynomial $r(n,s,k,d)$ such that, if
  $V_1,\dots,V_{2^k}\subseteq\cc^n$ are $2^k$ varieties, $V_i$ being defined
  by $2^s$ polynomials $f^{(i)}_1,\dots,f^{(i)}_{2^s}\in\zz[x_1,\dots,x_n]$ of
  degree $\leq 2^d$ and of coefficients bounded by $2^{2^M}$ in absolute
  value, then:
  \begin{enumerate}
  \item the variety $V=\cup_{i=1}^{2^k} V_i$ is defined by $n+1$ polynomials
    $g_1,\dots,g_{n+1}$ in $\zz[x_1,\dots,x_n]$ of degree $\leq 2^{d+k}$ and
    whose coefficients are bounded in absolute value by
    $2^{2^{M+r(n,s,k,d)}}$;
  \item moreover, the coefficients of the $g_i$ are bitwise computable from
    those of the $f_{j'}^{(j)}$ in space $Mr(n,s,k,d)$.
  \end{enumerate}
\end{proposition}
\begin{proof}
  We proceed by induction on $k$. Define
  $r(n,s,k,d)=(k+1)(p(n,s,d+k)+q(n,d+k))$, where $p$ and $q$ are defined in
  Proposition~\ref{prop_membership} and Lemma~\ref{lemma_union_two}
  respectively.  The base case $k=0$ is merely an application of
  Proposition~\ref{prop_membership}. For $k>0$, we first apply
  Proposition~\ref{prop_membership} to the $V_i$, so that each variety $V_i$
  is now defined by $n+1$ polynomials of degree $\leq 2^d$ and whose
  coefficients are bounded in absolute value by $2^{2^{M+p(n,s,d)}}$ and
  computable in space $Mp(n,s,d)$. Let us group the varieties $V_i$ by pairs:
  call $W_i=V_{2i-1}\cup V_{2i}$ for $1\leq i\leq 2^{k-1}$.  There are
  $2^{k-1}$ varieties $W_i$ and we have $V=\cup_iW_i$. By
  Lemma~\ref{lemma_union_two}, each variety $W_i$ is defined by $n+1$
  polynomials of degree $\leq 2^{d+1}$, of coefficients of bitsize
  $2^{M+p(n,s,d)+q(n,d)}$ and bitwise computable in space
  $M(p(n,s,d)+q(n,d))$. By induction hypothesis at rank $k-1$, $V$ is defined
  by $n+1$ polynomials of degree $\leq 2^{d+1+(k-1)}$, of coefficients of
  bitsize
  $2^{M+p(n,s,d)+q(n,d)+k(p(n,\lceil\log(n+1)\rceil,d+k-1)+q(n,d+k-1))}\leq
  2^{M+r(n,s,k,d)}$ and bitwise computable in space $Mr(n,s,k,d)$. This proves
  the proposition.\qed
\end{proof}

Here is the main consequence on membership tests to a union of
varieties.
\begin{corollary}\label{cor_union}
  Let $p(n)$ and $q(n)$ be two polynomials. Suppose $(f_n(\bar x,\bar
  y,\bar z))$ is a \unif\vpspacezero\ family with $|\bar x| = n$, $|\bar y|
  = p(n)$ and $|\bar z| = q(n)$. For an integer $0\leq i<2^{p(n)}$, call
  $V_i^{(n)}\subseteq\cc^n$ the variety defined by the polynomials $f_n(\bar
  x, i, j)$ for $0\leq j<2^{q(n)}$ (in this notation, $i$ and $j$ are encoded
  in binary).

  Then there exists a \unif\vpspacezero\ family $g_n(\bar x,\bar y,\bar z)$,
  where $|\bar x| = n$, $|\bar y| = p(n)$ and $|\bar z| =
  \lceil\log(n+1)\rceil$, such that
  $$\forall \bar x\in\cc^n,\ \ \forall k < 2^{p(n)},\ \ \left(
  \bar x\in \bigcup_{i=0}^kV_i^{(n)}\iff \bigwedge_{j=0}^ng_n(\bar x, k,
  j)=0\right).$$
\end{corollary}
\begin{proof}
  If $(f_n)$ is a \unif\vpspacezero\ family, by definition there exists a
  polynomial $p(n)$ such that the degree of $f_n$ is bounded by $2^{p(n)}$ and
  the absolute value of the coefficients by $2^{2^{p(n)}}$. Therefore $d$,
  $M$, $s$ and $k$ are polynomially bounded in Proposition~\ref{prop_union2}
  and the space needed to compute the coefficients of $g_n$ is polynomial.\qed
\end{proof}

\section{Proof of the Main Theorem}\label{sec_proof}

Sign conditions are the main ingredient of the proof.
Over \cc, we define the ``sign'' of $a\in\cc$ by 0 if $a=0$ and 1
otherwise. Let us fix a family of polynomials
$f_1,\ldots,f_s\in\zz[x_1,\ldots,x_n]$. A \emph{sign condition} is an element
$S\in\{0,1\}^s$. Hence there are $2^s$ sign conditions. Intuitively, the
$i$-th component of a sign condition determines the sign of the polynomial
$f_i$.

\subsection{Satisfiable Sign Conditions}

The sign condition of a point $\bar x\in\cc^n$ is the tuple $S^{\bar
  x}\in\{0,1\}^s$ defined by $S_i^{\bar x}=0\iff f_i(\bar x)=0$. We say that a
sign condition is \emph{satisfiable} if it is the sign condition of some
$\bar x\in\cc^n$. As 0-1 tuples, sign conditions can be viewed as subsets of
$\{1,\ldots,s\}$.  Using a fast parallel sorting algorithm (e.g. Cole's,
\cite{cole1988}), we can sort \ssc s in polylogarithmic parallel time in a way
compatible with set inclusion (e.g. the lexicographic order). We now fix such
a compatible linear order on sign conditions and consider our \ssc s
$S^{(1)}<S^{(2)}<\dots<S^{(N)}$ sorted accordingly.\label{linear_order}

The key point resides in the following theorem, coming from the algorithm
of~\cite{FGM1990}: there is a ``small'' number of \ssc s and enumerating them
is ``easy''.
\begin{theorem}\label{th_sc}
  Let $f_1,\ldots,f_s\in\zz[x_1,\ldots,x_n]$ and $d$ be their maximal
  degree. Then the number of \ssc s is $N=(sd)^{O(n)}$, and there is a
  uniform algorithm working in space $\bigl(n\log(sd)\bigr)^{O(1)}$ which, on
  boolean input $f_1,\ldots,f_s$ (in dense representation) and $(i,j)$ in
  binary, returns the $j$-th component of the $i$-th \ssc.
\end{theorem}

When $\log(sd)$ is polynomial in $n$, as will be the case, this yields a
\pspace\ algorithm. If furthermore the coefficients of $f_i$ are computable in
polynomial space, we will then be able to use the \ssc s in the
coefficients of \vpspace\ families, as in Lemma~\ref{cor_slice} below.

Let us explain why we are interested in sign conditions. An arithmetic circuit
performs tests of the form $f(\bar x)=0$ on input $\bar x\in\cc^n$, where $f$
is a polynomial. Suppose $f_1,\ldots,f_s$ is the list of all polynomials that
can be tested in \emph{any possible} computation. Then two elements of $\cc^n$
with the same sign condition are simultaneously accepted or rejected by the
circuit: the results of the tests are indeed always the same for both
elements.

Thus, instead of finding out whether $\bar x\in\cc^n$ is accepted by the
circuit, it is enough to find out whether the sign condition of $\bar x$ is
accepted. The advantage resides in handling only boolean tuples (the sign
conditions) instead of complex numbers (the input $\bar x$). But we have to be
able to find the sign condition of the input $\bar x$. This requires first the
enumeration of all the polynomials possibly tested in any computation of the
circuit.

\subsection{Enumerating all Possibly Tested Polynomials}

In the execution of an algebraic circuit, the values of some polynomials at
the input $\bar x$ are tested to zero.
In order to find the sign condition of the input $\bar x$, we have
to be able to enumerate in polynomial space all the polynomials that can ever
be tested to zero in the computations of an algebraic circuit. This is done
level by level as in~\cite[Th.~3]{CK1997} and~\cite{KP2006}.

\begin{proposition}\label{prop_slice}
  Let $C$ be a constant-free algebraic circuit with $n$ variables and of depth
  $d$.
  \begin{enumerate}
  \item The number of different polynomials possibly tested to zero in the
    computations of $C$ is $2^{d^2O(n)}$.
  \item There exists an algorithm using work space $(nd)^{O(1)}$ which, on
    input $C$ and integers $(i,j)$ in binary, outputs the $j$-th bit of the
    representation of the $i$-th polynomial.
  \end{enumerate}
\end{proposition}



Together with Theorem~\ref{th_sc}, this enables us to prove the following
result which will be useful in the proof of Proposition~\ref{prop_rank}: in
\unif\vpspacezero\ we can enumerate the polynomials as well as the \ssc s.

\begin{lemma}\label{cor_slice}
  Let $(C_n)$ be a uniform family of polynomial-depth algebraic circuits with
  polynomially many inputs. Call $d(n)$ the depth of $C_n$ and $i(n)$ the
  number of inputs. Let $f^{(n)}_1,\dots,f^{(n)}_s$ be all the polynomials
  possibly tested to zero by $C_n$ as in Proposition~\ref{prop_slice}, where
  $s=2^{O(nd(n)^2)}$. There are therefore $N=2^{O(n^2d(n)^2)}$ \ssc s
  $S^{(1)},\dots,S^{(N)}$ by Theorem~\ref{th_sc}.

  Then there exists a \unif\vpspacezero\ family $(g_n(\bar x,\bar y,\bar z))$,
  where $|\bar x| = i(n)$, $|\bar y| = O(n^2d(n)^2)$ and $|\bar z| =
  O(nd(n)^2)$, such that for all $1\leq i\leq N$ and $1\leq j\leq s$, we have:
  $$g_n(\bar x,i,j)=\left\{\begin{array}{ll}
      0 & \mbox{if } S^{(i)}_j=1\\
      f^{(n)}_j(\bar x) & \mbox{otherwise.}\end{array}\right.$$
\end{lemma}


\subsection{Finding the Sign Condition of the Input}

In order to find the sign condition $S^{\bar x}$ of the input $\bar
x\in\cc^n$, we will give a polynomial-time algorithm which tests some
\vpspace\ family for zero. Here is the formalized notion of a polynomial-time
algorithm with \vpspace\ tests.
\begin{definition}\label{def_algo_tests}
  A polynomial-time algorithm with \unif \vpspacezero\ tests is a \unif
  \vpspacezero\ family $(f_n(x_1,\dots,x_{u(n)}))$ together with a uniform
  family $(C_n)$ of constant-free polynomial-size algebraic circuits endowed
  with special test gates of indegree $u(n)$, whose value is $1$ on input
  $(a_1,\dots,a_{u(n)})$ if $f_n(a_1,\dots,a_{u(n)})=0$ and $0$ otherwise.
\end{definition}
Observe that a constant number of \unif\vpspacezero\ families can be used in the
preceding definition instead of only one: it is enough to combine them all in
one by using ``selection variables''. 

The precise result we show now is the following. By the ``rank'' of a
satisfiable sign condition, we merely mean its index in the fixed order on
satisfiable sign conditions.
\begin{proposition}\label{prop_rank}
  Let $(C_n)$ be a uniform family of algebraic circuits of polynomial depth
  and with a polynomial number $i(n)$ of inputs. There exists a
  polynomial-time algorithm with \unif\vpspacezero\ tests which, on input
  $\bar x\in\cc^{i(n)}$, returns the rank $i$ of the sign condition $S^{(i)}$ of
  $\bar x$ with respect to the polynomials $g_1,\dots,g_s$ tested to zero by
  $C_n$ given by Proposition~\ref{prop_slice}.
\end{proposition}

\begin{proof}
  Take the \unif\vpspacezero\ family $(g_n(\bar x,\bar y,\bar z))$ as in
  Lemma~\ref{cor_slice}: in essence, $g_n$ enumerates all the polynomials
  $f_1,\dots,f_s$ possibly tested to zero in $C_n$ and enumerates the $N$ \ssc
  s $S^{(1)}<\dots<S^{(N)}$. The idea now is to perform a binary search in
  order to find the rank $i$ of the sign condition of the input $\bar x$.

  Let $S^{(j)}\in\{0,1\}^s$ be a satisfiable sign condition. We say that
  $S^{(j)}$ is a \emph{candidate} whenever $\forall m\leq s$,
  $S^{(j)}_m=0\Rightarrow f_m(\bar x) = 0$. Remark that the sign condition of
  $\bar x$ is the smallest candidate. Call $V_j$ the variety defined by the
  polynomials $\{f_m | S^{(j)}_m=0\}$: by definition of $g_n$, $V_j$ is also
  defined by the polynomials $g_n(\bar x,j,k)$ for $k=1$ to $s$. Note that
  $S^{(j)}$ is a candidate if and only if $\bar x\in V_j$.

  Corollary~\ref{cor_union} combined with Lemma~\ref{cor_slice} asserts that
  tests of the form $\bar x\in \cup_{k\leq j}V_k$ are in
  \unif\vpspacezero. They are used to perform a binary search by making $j$
  vary. In a number of steps logarithmic in $N$ (i.e. polynomial in $n$), we
  find the rank $i$ of the sign condition of $\bar x$.
\qed
\end{proof}

\subsection{A Polynomial-time Algorithm for PAR$_{\cc}$ Problems}

\begin{lemma}\label{lem_accept}
  Let $(C_n)$ be a uniform family of constant-free polynomial-depth algebraic
  circuits. There is a (boolean) algorithm using work space polynomial in $n$
  which, on input $i$, decides whether the elements of the $i$-th satisfiable
  sign condition $S^{(i)}$ are accepted by the circuit $C_n$.
\end{lemma}

\begin{proof}
  We follow the circuit $C_n$ level by level. For test gates, we compute the
  polynomial $f$ to be tested. Then we enumerate the polynomials
  $f_1,\dots,f_s$ as in Proposition~\ref{prop_slice} for the circuit $C_n$ and
  we find the index $j$ of $f$ in this list. By consulting the $j$-th bit of
  the $i$-th satisfiable sign condition with respect to $f_1,\dots,f_s$ (which
  is done by the polynomial-space algorithm of Theorem~\ref{th_sc}), we
  therefore know the result of the test and can go on like this until the
  output gate.\qed
\end{proof}

\begin{theorem}
  Let $A\in\parczero$. There exists a polynomial-time algorithm with
  \unif\vpspacezero\ tests that decides $A$.
\end{theorem}

\begin{proof}
  $A$ is decided by a uniform family $(C_n)$ of constant-free polynomial-depth
  algebraic circuits. On input $\bar x$, thanks to Proposition~\ref{prop_rank}
  we first find the rank $i$ of the sign condition of $\bar x$ with respect to
  the polynomials $f_1,\dots,f_s$ of Proposition~\ref{prop_slice}. Then we
  conclude by a last \unif\vpspacezero\ test simulating the polynomial-space
  algorithm of Lemma~\ref{lem_accept} on input $i$.\qed
\end{proof}
Theorem~\ref{th_transfer} follows immediately from this result.
One could obtain other versions of these two results 
by changing the uniformity conditions or the role of constants.

\bibliographystyle{abbrv}
\bibliography{biblio}

\end{document}